%% file: main.tex
\documentclass{article}
\usepackage{amsmath, amssymb, amsthm, bbm, dsfont, lscape, mathtools, subcaption, xcolor, url}
\usepackage[toc,page]{appendix}
\usepackage[ruled,vlined]{algorithm2e}
\usepackage[bookmarks, bookmarksnumbered, linktocpage, urlcolor = cyan]{hyperref}
\usepackage{cleveref}
\AtBeginEnvironment{appendices}{\crefalias{section}{appendix}}

\newtheorem{lemma}{Lemma}
\newtheorem{proposition}{Proposition}
\newtheorem{theorem}{Theorem}

\newcommand{\footremember}[2]{%
    \footnote{#2}
    \newcounter{#1}
    \setcounter{#1}{\value{footnote}}%
}
\author{%
  Yizhou Cao\footremember{c}
  {Derivation Technology Ltd.,
  yizhou.cao@derivation.info}%
  \\\and Yepeng Ding\footremember{d}
  {Hiroshima University,
  yepengd@acm.org}%
  \\\and Ruichao Jiang \footremember{j}
  {Derivation Technology Ltd., ruichao.jiang@derivation.info}%
  \\\and Long Wen\footremember{w}
  {The Hong Kong Polytechnic University, long1wen@polyu.edu.hk}
}

\title{Chasing price drains liquidity}
\date{\today}

\begin{document}
    \maketitle
    \begin{abstract}
        Assuming that the price in a Uniswap v3 style Automated Market Maker (AMM) follows a Geometric Brownian Motion (GBM), we prove that the strategy that adjusts the position of liquidity to track the current price leads to a deterministic and exponentially fast decay of liquidity. Next, assuming that there is a Centralized Exchange (CEX), in which the price follows a GBM and the AMM price mean reverts to the CEX price, we show numerically that the same strategy still leads to decay. Last, we propose a strategy that increases the liquidity even without compounding fees earned through liquidity provision.
    \end{abstract}
    \input{introduction}
    \input{uni-v3}
    \input{chasing-strategy}
    \input{exogenous-price}
    \input{liquidity-dynamics}
    \input{safe-interval}
    \input{numeric}
    \input{discussion}
    \bibliographystyle{unsrt}
    \bibliography{biblio}
    \input{appendix}
\end{document}

%% file: introduction.tex
\section{Introduction} \label{sec:introduction}
AMMs allow Liquidity Providers (LPs) to provide liquidity passively to a trading pair X-Y of token X and token Y. Most AMMs are Constant Formula Market Maker (CFMM), i.e. all trades $(x+\Delta x,y+\Delta y)$, $\Delta x\Delta y<0$ must satisfy the constraint
\begin{equation} \label{eqn:cfmm}
    F(x+\Delta x,y+\Delta y)=F(x,y)
\end{equation}
for some function $F$ called the curve. CFMM quotes price of X in Y by implicit differentiation
\begin{equation}
    Z=-\frac{\mathrm{d}y}{\mathrm{d}x}=\frac{\partial F/\partial x}{\partial F/\partial y}.
\end{equation}
For example,Uniswap v2 uses curve $F(x,y)=xy$. Hence, the price of X in $Y$ in Uniswap v2 is $Z=\frac{y}{x}$, which satisfies a basic property: The less the X, the pricier it gets.

Most AMMs require that liquidity provision/withdrawal preserve the current price, i.e.
\begin{equation} \label{eqn:constnat-price}
    Z(x+\Delta x,y+\Delta y)=Z(x,y)
\end{equation}
for all $\Delta x\Delta y>0$.

For example, in Uniswap v2, \Cref{eqn:constnat-price} translates to $\Delta y=Z\Delta x$, i.e. equal value of X and Y be provided/withdrawn.

%% file: uni-v3.tex
\section{Range liquidity in Uniswap v3}
Consider a Uniswap v3 style AMM with the current price $Z$ of X in Y. To provide $L$ amount of liquidity over a price range $[Z_l,Z_r]$, the LP deposits
\begin{equation} \label{eqn:liquidity-provision}
    (X,Y)=\begin{cases}
        \left(L\left(\sqrt{\frac{1}{Z_l}}-\sqrt{\frac{1}{Z_r}}\right),0\right),& \text{if }Z<Z_l\\
        \left(L\left(\sqrt{\frac{1}{Z}}-\sqrt{\frac{1}{Z_r}}\right),L\left(\sqrt{Z}-\sqrt{Z_l}\right)\right),& \text{if }Z_l\leq Z\leq Z_r\\
        \left(0,L(\sqrt{Z_r}-\sqrt{{Z_l}})\right),& \text{if }Z_r<Z
\end{cases}
\end{equation}
amount of X and Y tokens, respectively \cite{milionis}.

\Cref{eqn:liquidity-provision} implies that an $l$ amount of liquidity over $(0,+\infty)$ can be decomposed as follows.
\begin{proposition}[Liquidity decomposition] \label{proposition:liquidity-decomposition}
    Let the current AMM price be $Z$ and a trade moves the price to $Z'$. Then for any interval $(a,b)\ni Z,Z'$, $l$ can be decomposed as three range liquidities with value $l$ over $(0,a)$, $(a,b)$, and $(b,\infty)$.
\end{proposition}
\begin{proof}
    There are two things to prove: First, the amount of underlying tokens are the same. Second, the effect of the trade over the whole price range $(0,+\infty)$ is identical to the trade over the active liquidity range $(a,b)$.
    
    The first part follows from
    \begin{equation*}
        \begin{split}
            x(l,Z,(0,\infty))&=\frac{l}{\sqrt{Z}}=0+l\left(\frac{1}{\sqrt{Z}}-\frac{1}{\sqrt{b}}\right)+l\left(\frac{1}{\sqrt{b}}-\frac{1}{\sqrt{\infty}}\right)\\
            &=x(l,Z,(0,a))+x(l,Z,(a,b))+x(l,Z,(b,\infty)),\\
            y(l,Z,(0,\infty))&=l\sqrt{Z}=l(\sqrt{a}-\sqrt{0})+l(\sqrt{Z}-\sqrt{a})+0\\
            &=y(l,Z,(0,a))+y(l,Z,(a,b))+y(l,Z,(b,\infty)).
        \end{split}
    \end{equation*}
    For the second part, assume WLOG $Z<Z'$. If the trade uses $l$ over $(0,+\infty)$, we have
    \begin{equation*}
        \begin{split}
            \Delta x&=\left(1-\sqrt{\frac{Z}{Z'}}\right)x(l,Z,(0,\infty)),\\
            \Delta y&=\left(\sqrt{\frac{Z'}{Z}}-1\right)y(l,Z,(0,\infty)).
        \end{split}
    \end{equation*}
    If the trade uses $l$ over $(a,b)$, we have
    \begin{equation*}
        \begin{split}
            \Delta x&=\left(1-\sqrt{\frac{Z}{Z'}}\right)\left(x(l,Z,(a,b))+\frac{l}{\sqrt{b}}\right)\\
            &=\left(1-\sqrt{\frac{Z}{Z'}}\right)\left[l\left(\frac{1}{\sqrt{Z}}-\frac{1}{\sqrt{b}}\right)+\frac{l}{\sqrt{b}}\right]\\
            &=\left(1-\sqrt{\frac{Z}{Z'}}\right)x(l,Z,(0,\infty)),\\
            \Delta y&=\left(\sqrt{\frac{Z'}{Z}}-1\right)\left(y(l,Z,(a,b))+l\sqrt{a}\right)\\
            &=\left(\sqrt{\frac{Z'}{Z}}-1\right)(l(\sqrt{Z}-\sqrt{a})+l\sqrt{a})\\
            &=\left(\sqrt{\frac{Z'}{Z}}-1\right)y(l,Z,(0,\infty)).
        \end{split}
    \end{equation*}
\end{proof}
\Cref{proposition:liquidity-decomposition} implies that overlapping range liquidities can be added.

%% file: chasing-strategy.tex
\section{Chasing liquidity dynamics} \label{sec:liquidity-process}
Assume that 
\begin{enumerate}
    \item $Z_t$ is continuous.
    \item For simplicity, at $t+dt$, $Z_{t+dt}$ always falls in $\left[\frac{Z_t}{\alpha},\alpha Z_t\right]$.
\end{enumerate}
Suppose that there is a background liquidity $l$ over $(0,+\infty)$. Consider the following continuous LP strategy.
\begin{enumerate}
    \item Withdraw $L_{t}$ liquidity over $\left[\frac{Z_t}{\alpha},\alpha Z_t\right]$ and obtain $L_{t}\left(\sqrt{\frac{1}{Z_{t+dt}}}-\sqrt{\frac{1}{\alpha Z_{t}}}\right)$ amount of X and $L_t\left(\sqrt{Z_{t+dt}}-\sqrt{\frac{Z_t}{\alpha}}\right)$ amount of Y.
    \item Provide liquidity over $\left[\frac{Z_{t+dt}}{\alpha},\alpha Z_{t+dt}\right]$, which is governed by the following three constraints from Uniswap v3
    \begin{equation*}
        \begin{split}
            &\left(x+\frac{l}{\sqrt{\alpha Z_{t+dt}}}\right)\left(y+l\sqrt{\frac{Z_{t+dt}}{\alpha}}\right)=l^2,\\
            &\left(x+\Delta X+\frac{l+L_{t+dt}}{\sqrt{\alpha Z_{t+dt}}}\right)\left[y+\Delta Y+(l+L_{t+dt})\sqrt{\frac{Z_{t+dt}}{\alpha}}\right]=(l+L_{t+dt})^2,\\
            &Z_{t+dt}=\frac{y+l\sqrt{\frac{Z_{t+dt}}{\alpha}}}{x+\frac{l}{\sqrt{\alpha Z_{t+dt}}}}=\frac{y+\Delta Y+(l+L_{t+dt})\sqrt{\frac{Z_{t+dt}}{\alpha}}}{x+\Delta x+\frac{l+L_{t+dt}}{\sqrt{\alpha Z_{t+dt}}}},
        \end{split}
    \end{equation*}
    and the following self-financing condition
    \begin{equation}
    \label{eqn:self-financing}
        \begin{split}
            \Delta X&=L_{t}\left(\sqrt{\frac{1}{Z_{t+dt}}}-\sqrt{\frac{1}{\alpha Z_t}}\right)+\delta X,\\
            \Delta Y&= L_{t}\left(\sqrt{Z_{t+dt}}-\sqrt{\frac{Z_t}{\alpha}}\right)+\delta Y,\\
            \delta Y&=-P_{t+dt}\delta X,
        \end{split}
    \end{equation}
\end{enumerate}
where $P_{t+dt}$ is the price of buying $\delta Y$ amount of Y at $t+dt$.

The solution is
\begin{equation} \label{eqn:liquidity-update}
    L_{t+dt}=\frac{\frac{P_{t+dt}}{\sqrt{Z_{t+dt}}}+\sqrt{Z_{t+dt}}-\sqrt{\frac{1}{\alpha}}\left(\frac{P_{t+dt}}{\sqrt{Z_t}}+\sqrt{Z_t}\right)}{\frac{P_{t+dt}}{\sqrt{Z_{t+dt}}}+\sqrt{Z_{t+dt}}}\frac{L_t}{1-\sqrt{\frac{1}{\alpha}}}.
\end{equation}

%% file: exogenous-price.tex
\section{Exogenous market model} \label{sec:exogenous-price}
Assume that
\begin{enumerate}
    \item The AMM price follows a GBM described by the following Stochastic Differential Equation (SDE)
        \begin{equation*}
        dZ_t=\mu Z_tdt+\sigma Z_tdW_t,
        \end{equation*}
        where $W_t$ is a standard Brownian motion.
    \item The exchange price coincides with the AMM price, i.e. $P_t=Z_t$.
\end{enumerate}
The assumption $P_t=Z_t$ seems contradictory because trading in AMM incurs slippage. But it is nevertheless necessary because otherwise $Z_t$ cannot be modeled exogenously as a GBM if the price slippage in AMM is taken into account. This assumption is justified when the price slippage is negligible when converting between X and Y. This model applies to tokens for which there exists no meaningful CEX.
\begin{theorem} \label{theorem:deterministic-decay}
    Under all assumptions in \Cref{sec:liquidity-process} and \Cref{sec:exogenous-price}, the liquidity process satisfies
    \begin{equation*}
        L_t=L_0\exp\left[-\frac{\sigma^2t}{8(\sqrt{\alpha}-1)}\right].
    \end{equation*}
\end{theorem}
\begin{proof}
    With $P_t=Z_t$, \Cref{eqn:liquidity-update} becomes
    \begin{equation*}
        L_{t+dt}=\frac{2-\sqrt{\frac{1}{\alpha}}\left(\sqrt{\frac{Z_{t+dt}}{Z_t}}+\sqrt{\frac{Z_t}{Z_{t+dt}}}\right)}{2}\frac{L_t}{1-\sqrt{\frac{1}{\alpha}}}.
    \end{equation*}
    Hence,
    \begin{equation*}
        \begin{split}
            dL_t&=\frac{d\sqrt{Z_t}\left(\sqrt{\frac{1}{Z_{t+dt}}}-\sqrt{\frac{1}{Z_{t}}}\right)}{2(\sqrt{\alpha}-1)}L_t\\
            &=\frac{d\sqrt{Z_t}d\sqrt{\frac{1}{Z_{t+dt}}}}{2(\sqrt{\alpha}-1)}L_t.
        \end{split}
    \end{equation*}
    By Ito's lemma,
    \begin{equation*}
        \begin{split}
            d\sqrt{Z_t}&=\frac{1}{2}Z_t^{-\frac{1}{2}}dZ_t-\frac{1}{8}Z_t^{-\frac{3}{2}}d\langle Z,Z\rangle t=\frac{1}{2}Z_t^{-\frac{1}{2}}dZ_t-\frac{\sigma^2}{8}\sqrt{Z_t}dt,\\
            d\sqrt{\frac{1}{Z_{t}}}&=-\frac{1}{2}Z_t^{-\frac{3}{2}}dZ_t+\frac{3}{8}Z_t^{-\frac{5}{2}}d\langle Z,Z\rangle t=-\frac{1}{2}Z_t^{-\frac{3}{2}}dZ_t+\frac{3\sigma^2}{8}Z_t^{-\frac{1}{2}}dt.
        \end{split}
    \end{equation*}
    Hence,
    \begin{equation*}
        \begin{split}
            dL_t&=\frac{\left(\frac{1}{2}Z_t^{-\frac{1}{2}}dZ_t-\frac{\sigma^2}{8}\sqrt{Z_t}dt\right)\left(-\frac{1}{2}Z_t^{-\frac{3}{2}}dZ_t+\frac{3\sigma^2}{8}Z_t^{-\frac{1}{2}}dt\right)}{2(\sqrt{\alpha}-1)}L_t\\
            &=-\frac{d\langle Z,Z\rangle t}{8(\sqrt{\alpha}-1)Z_t^2}L_t\\
            &=-\frac{\sigma^2}{8(\sqrt{\alpha}-1)}L_tdt.
        \end{split}
    \end{equation*}
\end{proof}
\Cref{theorem:deterministic-decay} says that the process $L$ is deterministic and $L$ decays exponentially with rate $\lambda=\frac{\sigma^2}{8(\sqrt{\alpha}-1)}$. So the higher the volatility and the more concentrated the liquidity ($\alpha\approx1^+$), the faster the liquidity that chases the current price decays.

%% file: liquidity-dynamics.tex
\section{Mean-reverting market model} \label{sec:mean-reverting}
Assume that 
\begin{enumerate}
    \item A CEX price $P_t$ follows a GBM
    \begin{equation*}
        dP_t=\mu P_tdt+\sigma P_tdW_t,
    \end{equation*}
    where $W_t$ is a standard Brownian motion.
    \item Following \cite{cartea}, the AMM price $Z_t$ is modeled by a mean reverting process
    \begin{equation*} \label{eqn:mean-reverting}
        dZ_t=\theta(P_t-Z_t)dt+\gamma Z_tdB_t,
    \end{equation*}
    where $\theta$ is the mean reversion speed parameter, $B_t$ is another standard Brownian motion, independent of $W_t$. 
\end{enumerate}
Note that
\begin{equation}
\label{eqn:independence-increment}
    dP_tdZ_t=\sigma\gamma P_tZ_tdW_tdB_t=0.
\end{equation}
\begin{theorem} \label{theorem:liquidity}
    The dynamics of $L_t$ is
    \begin{equation*}
        dL_t=\frac{\left({\frac{\gamma^2}{8}\frac{P_t^2}{Z_t^2}+\frac{3\gamma^2}{4}\frac{P_t}{Z_t}-\frac{3\gamma^2}{8}}\right)dt+\frac{1}{2}\left(1-\frac{P_t^2}{Z_t^2}\right)\left[\theta\left(\frac{P_t}{Z_t}-1\right)dt+\gamma dB_t\right]}{\left(1+\frac{P_t}{Z_t}\right)^2}\frac{L_t}{\sqrt{\alpha}-1}.
    \end{equation*}
\end{theorem}
The proof is in \Cref{appendix-a}.

Let 
\begin{equation*}
    \delta=\frac{P_t-Z_t}{Z_t}>-1
\end{equation*}
be the relative deviation of $P_t$ from $Z_t$, then
\begin{equation*}
    dL_t=\frac{-\frac{\theta}{2}\delta^3-\left(\theta-\frac{\gamma^2}{8}\right)\delta^2+\gamma^2\delta+\frac{\gamma^2}{2}}{(\sqrt{\alpha}-1)(\delta+2)^2}L_tdt-\frac{1}{2}\frac{\gamma\delta}{(\sqrt{\alpha}-1)(\delta+2)}L_tdB_t
\end{equation*}
Let
\begin{equation} \label{eqn:f_theta}
    f(\delta)\coloneqq-\frac{\theta}{2}\delta^3-\left(\theta-\frac{\gamma^2}{8}\right)\delta^2+\gamma^2\delta+\frac{\gamma^2}{2}.
\end{equation}
Then $f(0)=\frac{\gamma^2}{2}>0$, i.e. if $P_t=Z_t$ for all $t$, the drift is strictly positive. This is in contrast to \Cref{sec:exogenous-price}, in which we showed that $L_t$ decays exponentially by assuming $Z_t=P_t$. This is not a contradiction as in \Cref{sec:exogenous-price}, the exchange price $P_t$ was assumed to coincide with $Z_t$ almost surely, therefore $dP_tdZ_t\neq0$, whereas $dP_tdZ_t=0$ in this section.

%% file: safe-interval.tex
\section{Liquidity increasing strategy}
\label{sec:safe-interval}
\begin{lemma}
    $f(\delta)$ is strictly positive in some open neighbourhood $(\delta_l,\delta_r)\ni0$.
\end{lemma}
\begin{proof}
    Since
    \begin{equation*}
        \begin{split}
            &\lim_{\delta\to-\infty}f(\delta)=+\infty,\\
            &f(-1)=-\frac{\theta}{2}-\frac{3}{8}\gamma^2<0,\\
            &\lim_{\delta\to+\infty}f(\delta)=-\infty,
        \end{split}
    \end{equation*}
    $f(\delta)$ has one root in $(-\infty,-1)$, one root in $(-1,0)$, and one root in $(0,+\infty)$. Hence, $f(\delta)$ is strictly positive in an open neighbourhood $(\delta_l,\delta_r)\ni0$.
\end{proof}
If we provide liquidity only if $\delta\in(\delta_l,\delta_r)$, we expect to see an increasing in liquidity. However, doing so introduces jumps in $L_t$. To reconcile this, we use arbitrage to bring $\delta$ back to $(\delta_l,\delta_r)$. However, real arbitrage invalidates our assumption that $Z_t$ be continuous. So the following is only a heuristic.
\begin{enumerate}
    \item At time $t+dt$, we withdraw $L_{t}$ over $\left[\frac{Z_t}{\alpha},\alpha Z_t\right]$ to obtain $L_{t}\left(\sqrt{\frac{1}{Z_{t+dt}}}-\sqrt{\frac{1}{\alpha Z_{t}}}\right)$ amount of X and $L_t\left(\sqrt{Z_{t+dt}}-\sqrt{\frac{Z_t}{\alpha}}\right)$ amount of Y.
    \item If $\delta_{t+dt}\notin(\delta_l,\delta_r)$, we perform arbitrage so that $Z_{t+dt}\approx P_{t+dt}$.
    \item add range liquidity over $\left[\frac{P_{t+dt}}{\alpha},\alpha P_{t+dt}\right]$ subject to
    \begin{equation*}
        \begin{split}
            &\left(x+\frac{l}{\sqrt{\alpha P_{t+dt}}}\right)\left(y+l\sqrt{\frac{P_{t+dt}}{\alpha}}\right)=l^2,\\
            &\left(x+\Delta X+\frac{l+L_{t+dt}}{\sqrt{\alpha P_{t+dt}}}\right)\left[y+\Delta Y+(l+L_{t+dt})\sqrt{\frac{P_{t+dt}}{\alpha}}\right]=(l+L_{t+dt})^2,\\
            &P_{t+dt}=\frac{y+l\sqrt{\frac{P_{t+dt}}{\alpha}}}{x+\frac{l}{\sqrt{\alpha P_{t+dt}}}}=\frac{y+\Delta Y+(l+L_{t+dt})\sqrt{\frac{P_{t+dt}}{\alpha}}}{x+\Delta x+\frac{l+L_{t+dt}}{\sqrt{\alpha P_{t+dt}}}},
        \end{split}
    \end{equation*}
    and the self-financing condition \Cref{eqn:self-financing}.
    \end{enumerate}
We obtain the following update rule
\begin{equation} \label{eqn:add-liquidity}
    L_{t+dt}=\frac{\frac{P_{t+dt}}{\sqrt{Z_{t+dt}}}+\sqrt{Z_{t+dt}}-\sqrt{\frac{1}{\alpha}}\left(\frac{P_{t+dt}}{\sqrt{Z_t}}+\sqrt{Z_t}\right)}{2\sqrt{P_{t+dt}}}\frac{L_t}{1-\sqrt{\frac{1}{\alpha}}}.
\end{equation}
The liquidity provision strategy is summarized in \Cref{algo:liquidity-strategy}
\begin{algorithm}[ht!]
    \caption{Liquidity provision with arbitrage} \label{algo:liquidity-strategy}
    Withdraw $L_{t}$ over $\left[\frac{Z_t}{\alpha},\alpha Z_t\right]$\;
    \If{$\delta_l<\delta_{t+dt}<\delta_r$}{
        Add liquidity over $\left[\frac{Z_{t+dt}}{\alpha},\alpha Z_{t+dt}\right]$ according to \Cref{eqn:liquidity-update}.\;
    }
    \Else{
        Perform arbitrage so that $\delta_{t+dt}=0$\;
        Add liquidity over $\left[\frac{P_{t+dt}}{\alpha},\alpha P_{t+dt}\right]$ according to \Cref{eqn:add-liquidity}
    }
\end{algorithm}

The next task is to determine $\delta_l$ and $\delta_r$. In principle, there is no difficulty as $f(\delta)$ is a cubic polynomial, for which explicit formula exists for its roots but it hinders the relationship between the parameters.

Assume that
\begin{enumerate}
    \item $\gamma^2\ll\theta$. Otherwise, the mean reversion process isn't a good approximation to arbitrageurs.
    \item $\delta\ll\theta$ as we work within a small neighborhood of $0$.
\end{enumerate}
So we drop $-\frac{\theta}{2}\delta^3$ and $\frac{\gamma^2}{8}\delta^2$ in \Cref{eqn:f_theta}:
\begin{equation*}
    f(\delta)\approx-\theta\delta^2+\gamma^2\delta+\frac{\gamma^2}{2}.
\end{equation*}
And the roots are
\begin{equation*}
    \begin{split}
        \delta_l&=\frac{\gamma^2}{2\theta}-\frac{\gamma}{\sqrt{2\theta}}<0,\\
        \delta_r&=\frac{\gamma^2}{2\theta}+\frac{\gamma}{\sqrt{2\theta}}>0.
    \end{split}
\end{equation*}
Hence, as long as
\begin{equation} \label{eqn:range}
    \frac{P_t}{1+\frac{\gamma}{\sqrt{2\theta}}+\frac{\gamma^2}{2\theta}}<Z_t<\frac{P_t}{1-\frac{\gamma}{\sqrt{2\theta}}+\frac{\gamma^2}{2\theta}},
\end{equation}
the liquidity often tends to increase, which is confirmed by simulation in \Cref{sec:numerics}.

Moreover, if the arbitrage intensity $\theta\to+\infty$, \Cref{eqn:range} becomes empty, which is consistent with \Cref{sec:exogenous-price}.

%% file: numeric.tex
\section{Numerical result} \label{sec:numerics}
We used Binance ETH-USDC pair from 2024-01-15 to 2024-9-15 to estimate the parameters of $P_t$ and Uniswap v3 on Base blockchain from 2024-02-01 to 2024-07-18 to estimate the parameters of $Z_t$. The estimators are derived in \Cref{sec:estimators}. The results are
\begin{align*}
    \hat{\mu}&=-1.17,\\
    \hat{\sigma}&=0.75,\\
    \hat{\theta}&=1058.49,\\
    \hat{\gamma}&=0.68,
\end{align*}
units in per year. Hence,
\begin{equation*}
    \begin{split}
        \delta_l&\approx-0.014,\\ \delta_r&\approx0.015.
    \end{split}
\end{equation*}
We perform Monte Carlo simulation with initial price 2000, initial liquidity 1000 with $\alpha=1.1$. The simulation contains 1000 rounds and each round lasts 35280 time steps, with step size being 1 min.
\begin{figure}[ht]
    \centering
    \includegraphics[width=\linewidth]{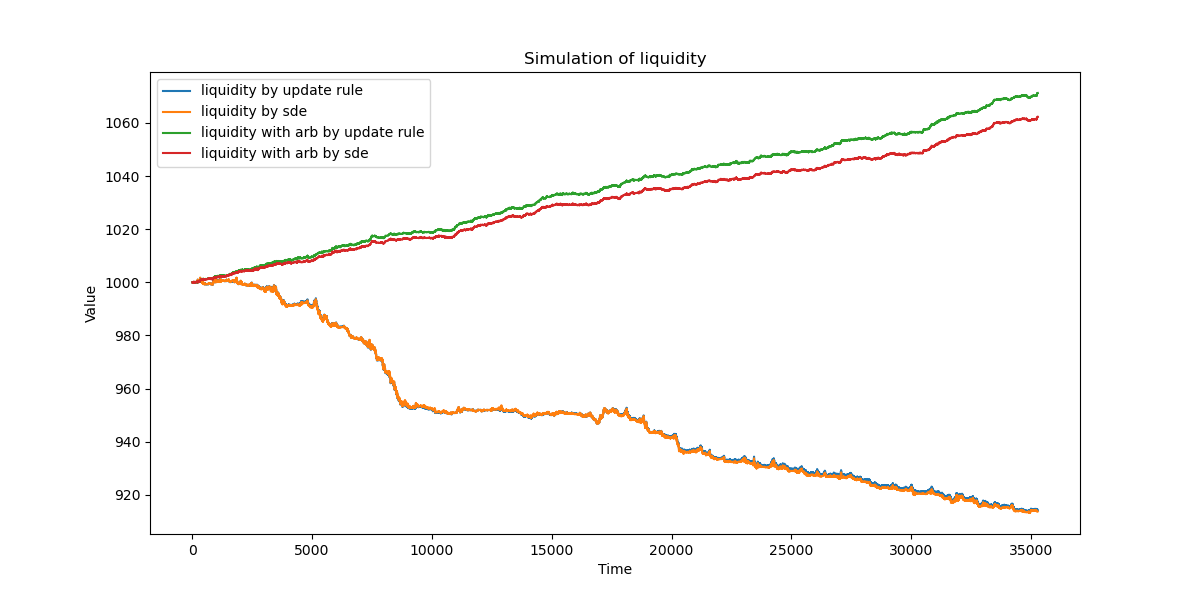}
    \caption{The blue line corresponds to \Cref{eqn:liquidity-update}, orange and red to \Cref{theorem:liquidity}, and green to \Cref{eqn:add-liquidity}.}
    \label{fig:simulation}
\end{figure}

\Cref{fig:simulation} shows that without arbitrage, the liquidity still decays, but the decay is not deterministic. The fact that the blue and the orange lines coincide shows that our derived SDE of $L_t$ (\Cref{theorem:liquidity}) is accurate. With arbitrage, the liquidity increases. However, the discrepancy between the green and the red lines shows that \Cref{theorem:liquidity} becomes inaccurate as performing arbitrage invalidates the continuous AMM price assumption.

%% file: discussion.tex
\section{Conclusion}
In this paper, we derived the SDE of the liquidity process induced by the strategy that chases the current price in a Uniswap v3 style AMM under two market models. If the AMM price is modeled as a GBM, we proved that the liquidity decays deterministically and exponentially fast. If the AMM price is modelled as a mean-reverting process, the numerical simulation showed that the liquidity still decays. However, if we provide liquidity according to \Cref{algo:liquidity-strategy}, the numerical simulation showed an increase of liquidity, even without taking fees and potential profit from the arbitrage into account.

%% file: appendix.tex
\begin{appendices}
\section{Proof of \texorpdfstring{\Cref{theorem:liquidity}}{\ref{theorem:liquidity}}} \label{appendix-a}
We evaluate two derivatives first.
\begin{lemma}
\label{lemma:derivative}
    \begin{equation*}
        \begin{split}
            &d\frac{1}{\frac{P_t}{\sqrt{Z_t}}+\sqrt{Z_t}}\\
            =&\frac{\left(\frac{\sigma^2P_t^2}{Z_t}-\frac{\gamma^2}{8}\frac{P_t^2}{Z_t}-\frac{3\gamma^2}{4}P_t+\frac{3\gamma^2}{8}Z_t\right)dt-\left(1+\frac{P_t}{Z_t}\right)dP_t+\frac{1}{2}\left(1-\frac{P_t^2}{Z_t^2}\right)dZ_t}{\left(\frac{P_t}{\sqrt{Z_t}}+\sqrt{Z_t}\right)^3}.
        \end{split}
    \end{equation*}
\end{lemma}
\begin{proof}
    By \Cref{eqn:independence-increment},
    \begin{equation*}
        \begin{split}
            d\left(\frac{P_t}{\sqrt{Z_t}}+\sqrt{Z_t}\right)&=\frac{dP_t}{\sqrt{Z_t}}+P_td\frac{1}{\sqrt{Z_t}}+dP_td\frac{1}{\sqrt{Z_t}}+d\sqrt{Z_t}\\
            &=\frac{dP_t}{\sqrt{Z_t}}+P_td\frac{1}{\sqrt{Z_t}}+d\sqrt{Z_t}.
        \end{split}
    \end{equation*}
    Then
    \begin{equation*}
        \begin{split}
            &d\left(\frac{P_t}{\sqrt{Z_t}}+\sqrt{Z_t}\right)d\left(\frac{P_t}{\sqrt{Z_t}}+\sqrt{Z_t}\right)\\
            =&\frac{(dP_t)^2}{Z_t}+P_t^2\left(d\frac{1}{\sqrt{Z_t}}\right)^2+2P_td\sqrt{Z_t}d\frac{1}{\sqrt{Z_t}}+(d\sqrt{Z_t})^2\\
            =&\left(\frac{\sigma^2P_t^2}{Z_t}+\frac{\gamma^2}{4}\frac{P_t^2}{Z_t}-\frac{\gamma^2}{2}P_t+\frac{\gamma^2}{4}Z_t\right)dt.
        \end{split}
    \end{equation*}
    Also,
    \begin{equation*}
        \begin{split}
            &\left(\frac{P_t}{\sqrt{Z_t}}+\sqrt{Z_t}\right)d\left(\frac{P_t}{\sqrt{Z_t}}+\sqrt{Z_t}\right)\\
            =&\left(\frac{P_t}{\sqrt{Z_t}}+\sqrt{Z_t}\right)\left(\frac{dP_t}{\sqrt{Z_t}}+P_td\frac{1}{\sqrt{Z_t}}+d\sqrt{Z_t}\right)\\
            =&\frac{P_t}{Z_t}dP_t+\frac{P_t^2}{\sqrt{Z_t}}\left(-\frac{1}{2}Z_t^{-\frac{3}{2}}dZ_t+\frac{3\gamma^2}{8}Z_t^{-\frac{1}{2}}dt\right)+\frac{P_t}{\sqrt{Z_t}}\left(\frac{1}{2}Z_t^{-\frac{1}{2}}dZ_t\right.\\
            &\left.-\frac{\gamma^2}{8}\sqrt{Z_t}dt\right)+dP_t+P_t\sqrt{Z_t}\left(-\frac{1}{2}Z_t^{-\frac{3}{2}}dZ_t+\frac{3\gamma^2}{8}Z_t^{-\frac{1}{2}}dt\right)\\
            &+\sqrt{Z_t}\left(\frac{1}{2}Z_t^{-\frac{1}{2}}dZ_t-\frac{\gamma^2}{8}\sqrt{Z_t}dt\right)\\
            =&\frac{P_t}{Z_t}dP_t-\frac{1}{2}\frac{P_t^2}{Z_t^2}dZ_t+\frac{3\gamma^2}{8}\frac{P_t^2}{Z_t}dt+\frac{1}{2}\frac{P_t}{Z_t}dZ_t-\frac{\gamma^2}{8}P_tdt+dP_t\\
            &-\frac{1}{2}\frac{P_t}{Z_t}dZ_t+\frac{3\gamma^2}{8}P_tdt+\frac{1}{2}dZ_t-\frac{\gamma^2}{8}Z_tdt\\
            =&\left(1+\frac{P_t}{Z_t}\right)dP_t+\frac{1}{2}\left(1-\frac{P_t^2}{Z_t^2}\right)dZ_t+\frac{\gamma^2}{8}\left(\frac{3P_t^2}{Z_t}+2P_t-Z_t\right)dt
        \end{split}
    \end{equation*}
    By the quotient rule for stochastic process,
    {\allowdisplaybreaks
    \begin{align*}
        &d\frac{1}{\frac{P_t}{\sqrt{Z_t}}+\sqrt{Z_t}}\\
        =&\frac{1}{\frac{P_t}{\sqrt{Z_t}}+\sqrt{Z_t}}\left[-\frac{d\left(\frac{P_t}{\sqrt{Z_t}}+\sqrt{Z_t}\right)}{\frac{P_t}{\sqrt{Z_t}}+\sqrt{Z_t}}+\frac{d\left(\frac{P_t}{\sqrt{Z_t}}+\sqrt{Z_t}\right)d\left(\frac{P_t}{\sqrt{Z_t}}+\sqrt{Z_t}\right)}{\left(\frac{P_t}{\sqrt{Z_t}}+\sqrt{Z_t}\right)^2}\right]\\
        =&\frac{d\left(\frac{P_t}{\sqrt{Z_t}}+\sqrt{Z_t}\right)d\left(\frac{P_t}{\sqrt{Z_t}}+\sqrt{Z_t}\right)-\left(\frac{P_t}{\sqrt{Z_t}}+\sqrt{Z_t}\right)d\left(\frac{P_t}{\sqrt{Z_t}}+\sqrt{Z_t}\right)}{\left(\frac{P_t}{\sqrt{Z_t}}+\sqrt{Z_t}\right)^3}\\
        =&\frac{\left(\frac{\sigma^2P_t^2}{Z_t}+\frac{\gamma^2}{4}\frac{P_t^2}{Z_t}-\frac{\gamma^2}{2}P_t+\frac{\gamma^2}{4}Z_t\right)dt-\frac{\gamma^2}{8}\left(\frac{3P_t^2}{Z_t}+2P_t-Z_t\right)dt}{\left(\frac{P_t}{\sqrt{Z_t}}+\sqrt{Z_t}\right)^3}\\
        &-\frac{\left(1+\frac{P_t}{Z_t}\right)dP_t+\frac{1}{2}\left(1-\frac{P_t^2}{Z_t^2}\right)dZ_t}{\left(\frac{P_t}{\sqrt{Z_t}}+\sqrt{Z_t}\right)^3}\\
        =&\frac{\left(\frac{\sigma^2P_t^2}{Z_t}-\frac{\gamma^2}{8}\frac{P_t^2}{Z_t}-\frac{3\gamma^2}{4}P_t+\frac{3\gamma^2}{8}Z_t\right)dt-\left(1+\frac{P_t}{Z_t}\right)dP_t-\frac{1}{2}\left(1-\frac{P_t^2}{Z_t^2}\right)dZ_t}{\left(\frac{P_t}{\sqrt{Z_t}}+\sqrt{Z_t}\right)^3}.
    \end{align*}
    }
\end{proof}
\begin{lemma}
\label{lemma:derivative-2}
    \begin{equation*}
        \begin{split}
            d\frac{1}{\left(\frac{P_t}{\sqrt{Z_t}}+\sqrt{Z_t}\right)^2}&=\frac{4\sigma^2\left(1+\frac{P_t}{Z_t}\right)^2P_t^2dt+\gamma^2\left(1-\frac{P_t^2}{Z_t^2}\right)Z_t^2dt}{\left(\frac{P_t}{\sqrt{Z_t}}+\sqrt{Z_t}\right)^6}\\
            &-\frac{(\sigma^2+\gamma^2)\frac{P_t^2}{Z_t}dt+2\left(1+\frac{P_t}{Z_t}\right)dP_t+\left(1-\frac{P_t^2}{Z_t^2}\right)dZ_t}{\left(\frac{P_t}{\sqrt{Z_t}}+\sqrt{Z_t}\right)^4}
        \end{split}
    \end{equation*}
\end{lemma}
\begin{proof}
    By the quotient rule for stochastic process,
    \begin{equation*}
        d\frac{1}{\left(\frac{P_t}{\sqrt{Z_t}}+\sqrt{Z_t}\right)^2}=-\frac{d\left(\frac{P_t}{\sqrt{Z_t}}+\sqrt{Z_t}\right)^2}{\left(\frac{P_t}{\sqrt{Z_t}}+\sqrt{Z_t}\right)^4}+\frac{d\left(\frac{P_t}{\sqrt{Z_t}}+\sqrt{Z_t}\right)^2d\left(\frac{P_t}{\sqrt{Z_t}}+\sqrt{Z_t}\right)^2}{\left(\frac{P_t}{\sqrt{Z_t}}+\sqrt{Z_t}\right)^6}
    \end{equation*}
    We evaluate
    {\allowdisplaybreaks
    \begin{align*}
        &d\left(\frac{P_t}{\sqrt{Z_t}}+\sqrt{Z_t}\right)^2\\
        =&2\left(\frac{P_t}{\sqrt{Z_t}}+\sqrt{Z_t}\right)d\left(\frac{P_t}{\sqrt{Z_t}}+\sqrt{Z_t}\right)+d\left(\frac{P_t}{\sqrt{Z_t}}+\sqrt{Z_t}\right)d\left(\frac{P_t}{\sqrt{Z_t}}+\sqrt{Z_t}\right)\\
        =&2\left(1+\frac{P_t}{Z_t}\right)dP_t+\left(1-\frac{P_t^2}{Z_t^2}\right)dZ_t+\frac{\gamma^2}{4}\left(\frac{3P_t^2}{Z_t}+2P_t-Z_t\right)dt\\
        &+\left(\frac{\sigma^2P_t^2}{Z_t}+\frac{\gamma^2}{4}\frac{P_t^2}{Z_t}-\frac{\gamma^2}{2}P_t+\frac{\gamma^2}{4}Z_t\right)dt\\
        =&(\sigma^2+\gamma^2)\frac{P_t^2}{Z_t}dt+2\left(1+\frac{P_t}{Z_t}\right)dP_t+\left(1-\frac{P_t^2}{Z_t^2}\right)dZ_t.
    \end{align*}
    }
    Then
    {\allowdisplaybreaks
    \begin{align*}
        &d\left(\frac{P_t}{\sqrt{Z_t}}+\sqrt{Z_t}\right)^2d\left(\frac{P_t}{\sqrt{Z_t}}+\sqrt{Z_t}\right)^2\\
        =&\left[4\sigma^2\left(1+\frac{P_t}{Z_t}\right)^2P_t^2+\gamma^2\left(1-\frac{P_t^2}{Z_t^2}\right)^2Z_t^2\right]dt.
    \end{align*}
    }
\end{proof}
The proof of \Cref{theorem:liquidity} is as follows.
\begin{proof}
    By \Cref{eqn:liquidity-update}
    {\allowdisplaybreaks
    \begin{align*}
        dL_t=&\frac{\frac{P_{t+dt}}{\sqrt{Z_{t+dt}}}-\frac{P_{t+dt}}{\sqrt{Z_t}}+\sqrt{Z_{t+dt}}-\sqrt{Z_t}}{\frac{P_{t+dt}}{\sqrt{Z_{t+dt}}}+\sqrt{Z_{t+dt}}}\frac{L_t}{\sqrt{\alpha}-1}\\
        =&\frac{L_t}{\sqrt{\alpha}-1}-\frac{\frac{P_{t+dt}}{\sqrt{Z_t}}+\sqrt{Z_t}}{\frac{P_{t+dt}}{\sqrt{Z_{t+dt}}}+\sqrt{Z_{t+dt}}}\frac{L_t}{\sqrt{\alpha}-1}\\
        =&\frac{L_t}{\sqrt{\alpha}-1}-\frac{\frac{P_t}{\sqrt{Z_t}}+\sqrt{Z_t}}{\frac{P_{t+dt}}{\sqrt{Z_{t+dt}}}+\sqrt{Z_{t+dt}}}\frac{L_t}{\sqrt{\alpha}-1}-\frac{\frac{dP_t}{\sqrt{Z_t}}}{\frac{P_{t+dt}}{\sqrt{Z_{t+dt}}}+\sqrt{Z_{t+dt}}}\frac{L_t}{\sqrt{\alpha}-1}\\
        =&\left(\frac{P_t}{\sqrt{Z_t}}+\sqrt{Z_t}\right)\left(\frac{1}{\frac{P_t}{\sqrt{Z_t}}+\sqrt{Z_t}}-\frac{1}{\frac{P_{t+dt}}{\sqrt{Z_{t+dt}}}+\sqrt{Z_{t+dt}}}\right)\frac{L_t}{\sqrt{\alpha}-1}\\
        &-\frac{\frac{dP_t}{\sqrt{Z_t}}}{\frac{P_{t+dt}}{\sqrt{Z_{t+dt}}}+\sqrt{Z_{t+dt}}}\frac{L_t}{\sqrt{\alpha}-1}\\
        =&-\left(\frac{P_t}{\sqrt{Z_t}}+\sqrt{Z_t}\right)d\left(\frac{1}{\frac{P_t}{\sqrt{Z_t}}+\sqrt{Z_t}}\right)\frac{L_t}{\sqrt{\alpha}-1}-\frac{\frac{dP_t}{\sqrt{Z_t}}}{\frac{P_{t+dt}}{\sqrt{Z_{t+dt}}}+\sqrt{Z_{t+dt}}}\frac{L_t}{\sqrt{\alpha}-1}\\
        =&\left[\frac{\left({\frac{\gamma^2}{8}\frac{P_t^2}{Z_t}+\frac{3\gamma^2}{4}P_t-\frac{3\gamma^2}{8}Z_t-\sigma^2\frac{P_t^2}{Z_t}}\right)dt+\left(1+\frac{P_t}{Z_t}\right)dP_t+\frac{1}{2}\left(1-\frac{P_t^2}{Z_t^2}\right)dZ_t}{\left(\frac{P_t}{\sqrt{Z_t}}+\sqrt{Z_t}\right)^2}\right.\\
        &-\left.\frac{\frac{dP_t}{\sqrt{Z_t}}}{\frac{P_{t+dt}}{\sqrt{Z_{t+dt}}}+\sqrt{Z_{t+dt}}}\right]\frac{L_t}{\sqrt{\alpha}-1}.
    \end{align*}
    }
    We calculate
    {\allowdisplaybreaks
    \begin{align*}
        &\frac{\left(1+\frac{P_t}{Z_t}\right)dP_t}{\left(\frac{P_t}{\sqrt{Z_t}}+\sqrt{Z_t}\right)^2}-\frac{\frac{dP_t}{\sqrt{Z_t}}}{\frac{P_{t+dt}}{\sqrt{Z_{t+dt}}}+\sqrt{Z_{t+dt}}}\\
        =&\frac{\left(1+\frac{P_t}{Z_t}\right)dP_t}{\left(\frac{P_t}{\sqrt{Z_t}}+\sqrt{Z_t}\right)^2}-\frac{\frac{dP_t}{\sqrt{Z_t}}\left(\frac{P_{t+dt}}{\sqrt{Z_{t+dt}}}+\sqrt{Z_{t+dt}}\right)}{\left(\frac{P_{t+dt}}{\sqrt{Z_{t+dt}}}+\sqrt{Z_{t+dt}}\right)^2}\\
        =&\frac{\left(1+\frac{P_t}{Z_t}\right)dP_t}{\left(\frac{P_t}{\sqrt{Z_t}}+\sqrt{Z_t}\right)^2}-\frac{\frac{dP_t}{\sqrt{Z_t}}\left[(P_t+dP_t)\left(\frac{1}{\sqrt{Z_t}}+d\frac{1}{\sqrt{Z_t}}\right)+\sqrt{Z_t}+d\sqrt{Z_t}\right]}{\left(\frac{P_{t+dt}}{\sqrt{Z_{t+dt}}}+\sqrt{Z_{t+dt}}\right)^2}\\
        =&\frac{\left(1+\frac{P_t}{Z_t}\right)dP_t}{\left(\frac{P_t}{\sqrt{Z_t}}+\sqrt{Z_t}\right)^2}-\frac{\frac{dP_t}{\sqrt{Z_t}}\left(\frac{P_t}{\sqrt{Z_t}}+\frac{dP_t}{\sqrt{Z_t}}+\sqrt{Z_t}\right)}{\left(\frac{P_{t+dt}}{\sqrt{Z_{t+dt}}}+\sqrt{Z_{t+dt}}\right)^2}\\
        =&\frac{\left(1+\frac{P_t}{Z_t}\right)dP_t}{\left(\frac{P_t}{\sqrt{Z_t}}+\sqrt{Z_t}\right)^2}-\frac{\left(1+\frac{P_t}{Z_t}\right)dP_t}{\left(\frac{P_{t+dt}}{\sqrt{Z_{t+dt}}}+\sqrt{Z_{t+dt}}\right)^2}-\frac{\sigma^2\frac{P_t^2}{Z_t}dt}{\left(\frac{P_{t+dt}}{\sqrt{Z_{t+dt}}}+\sqrt{Z_{t+dt}}\right)^2}\\
        =&-\left(1+\frac{P_t}{Z_t}\right)dP_td\frac{1}{\left(\frac{P_t}{\sqrt{Z_t}}+\sqrt{Z_t}\right)^2}-\frac{\sigma^2\frac{P_t^2}{Z_t}dt}{\left(\frac{P_{t+dt}}{\sqrt{Z_{t+dt}}}+\sqrt{Z_{t+dt}}\right)^2}\\
        =&-\left(1+\frac{P_t}{Z_t}\right)dP_t\left[\frac{4\sigma^2\left(1+\frac{P_t}{Z_t}\right)^2P_t^2dt+\gamma^2\left(1-\frac{P_t^2}{Z_t^2}\right)Z_t^2dt}{\left(\frac{P_t}{\sqrt{Z_t}}+\sqrt{Z_t}\right)^6}\right.\\
        &-\left.\frac{(\sigma^2+\gamma^2)\frac{P_t^2}{Z_t}dt+2\left(1+\frac{P_t}{Z_t}\right)dP_t+\left(1-\frac{P_t^2}{Z_t^2}\right)dZ_t}{\left(\frac{P_t}{\sqrt{Z_t}}+\sqrt{Z_t}\right)^4}\right]\\
        &-\frac{\sigma^2\frac{P_t^2}{Z_t}dt}{\left(\frac{P_{t+dt}}{\sqrt{Z_{t+dt}}}+\sqrt{Z_{t+dt}}\right)^2}\\
        =&2\sigma^2\frac{\left(1+\frac{P_t}{Z_t}\right)^2P_t^2}{\left(1+\frac{P_t}{Z_t}\right)^4Z_t^2}dt-\frac{\sigma^2\frac{P_t^2}{Z_t}}{\left(\frac{P_{t+dt}}{\sqrt{Z_{t+dt}}}+\sqrt{Z_{t+dt}}\right)^2}dt\\
        =&\frac{\sigma^2\frac{P_t^2}{Z_t}}{\left(\frac{P_t}{\sqrt{Z_t}}+\sqrt{Z_t}\right)^2}dt+\sigma^2\frac{P_t^2}{Z_t}\left[\frac{1}{\left(\frac{P_t}{\sqrt{Z_t}}+\sqrt{Z_t}\right)^2}-\frac{1}{\left(\frac{P_{t+dt}}{\sqrt{Z_{t+dt}}}+\sqrt{Z_{t+dt}}\right)^2}\right]dt\\
        =&\frac{\sigma^2\frac{P_t^2}{Z_t}}{\left(\frac{P_t}{\sqrt{Z_t}}+\sqrt{Z_t}\right)^2}dt-\sigma^2\frac{P_t^2}{Z_t}d\frac{1}{\left(\frac{P_t}{\sqrt{Z_t}}+\sqrt{Z_t}\right)^2}dt\\
        =&\frac{\sigma^2\frac{P_t^2}{Z_t}}{\left(\frac{P_t}{\sqrt{Z_t}}+\sqrt{Z_t}\right)^2}dt.
    \end{align*}
    }
    Hence,
    {\allowdisplaybreaks
    \begin{align*}
        dL_t=&\left[\frac{\left({\frac{\gamma^2}{8}\frac{P_t^2}{Z_t}+\frac{3\gamma^2}{4}P_t-\frac{3\gamma^2}{8}Z_t-\sigma^2\frac{P_t^2}{Z_t}}\right)dt+\frac{1}{2}\left(1-\frac{P_t^2}{Z_t^2}\right)dZ_t}{\left(\frac{P_t}{\sqrt{Z_t}}+\sqrt{Z_t}\right)^2}\right.\\
        &+\left.\frac{\sigma^2\frac{P_t^2}{Z_t}}{\left(\frac{P_t}{\sqrt{Z_t}}+\sqrt{Z_t}\right)^2}dt\right]\frac{L_t}{\sqrt{\alpha}-1}\\
        =&\frac{\left({\frac{\gamma^2}{8}\frac{P_t^2}{Z_t}+\frac{3\gamma^2}{4}P_t-\frac{3\gamma^2}{8}Z_t}\right)dt+\frac{1}{2}\left(1-\frac{P_t^2}{Z_t^2}\right)dZ_t}{\left(\frac{P_t}{\sqrt{Z_t}}+\sqrt{Z_t}\right)^2}\frac{L_t}{\sqrt{\alpha}-1}\\
        =&\frac{\left({\frac{\gamma^2}{8}\frac{P_t^2}{Z_t}+\frac{3\gamma^2}{4}P_t-\frac{3\gamma^2}{8}Z_t}\right)dt+\frac{1}{2}\left(1-\frac{P_t^2}{Z_t^2}\right)[\theta(P_t-Z_t)dt+\gamma Z_tdB_t]}{\left(\frac{P_t}{\sqrt{Z_t}}+\sqrt{Z_t}\right)^2}\\
        &\frac{L_t}{\sqrt{\alpha}-1}\\
        =&\frac{\left({\frac{\gamma^2}{8}\frac{P_t^2}{Z_t^2}+\frac{3\gamma^2}{4}\frac{P_t}{Z_t}-\frac{3\gamma^2}{8}}\right)dt+\frac{1}{2}\left(1-\frac{P_t^2}{Z_t^2}\right)\left[\theta\left(\frac{P_t}{Z_t}-1\right)dt+\gamma dB_t\right]}{\left(1+\frac{P_t}{Z_t}\right)^2}\\
        &\frac{L_t}{\sqrt{\alpha}-1}
    \end{align*}
    }
\end{proof}
\section{Statistical estimators} \label{sec:estimators}
We estimate the parameters $(\mu,\sigma,\theta,\gamma)$ in
\begin{equation*}
    \begin{split}
        &dP_t=\mu P_tdt+\sigma P_tdW_t,\\
        &dZ_t=\theta(P_t-Z_t)dt+\gamma Z_tdB_t.
    \end{split}
\end{equation*}
Suppose that we observe $(P_{t_i},Z_{t_i})$ at $i=0,1,2,\cdots,N$ and $\Delta T=t_{i+1}-t_i$ for all $i$.

The geometric Brownian motion has explicit solution
\begin{equation*}
    P_t=P_0\exp\left\{\left(\mu-\frac{1}{2}\sigma^2\right)t+\sigma W_t\right\}.
\end{equation*}
Therefore, we have unbiased estimators
\begin{equation*}
    \begin{split}
        &\frac{1}{N}\sum_{i=0}^{N-1}\ln\frac{P_{t_{i+1}}}{P_{t_i}}\to\left(\mu-\frac{1}{2}\sigma^2\right)\Delta T,\\
        &\frac{1}{N-1}\sum_{i=0}^{N-1}\left(\ln\frac{P_{t_{i+1}}}{P_{t_i}}-\frac{1}{N}\sum_{i=0}^{N-1}\ln\frac{P_{t_{i+1}}}{P_{t_i}}\right)^2\to\sigma^2\Delta T.
    \end{split}
\end{equation*}
Hence, the unbiased estimators for $\mu$ and $\sigma$ are
\begin{equation}
    \begin{split}
        \hat{\mu}=&\frac{1}{N\Delta T}\sum_{i=0}^{N-1}\ln\frac{P_{t_{i+1}}}{P_{t_i}}+\frac{1}{2(N-1)\Delta T}\left[\sum_{i=0}^{N-1}\left(\ln\frac{P_{t_{i+1}}}{P_{t_i}}\right)^2\right.\\
        &\left.-\frac{1}{N}\left(\sum_{i=0}^{N-1}\ln\frac{P_{t_{i+1}}}{P_{t_i}}\right)^2\right],\\
        \hat{\sigma}^2=&\frac{1}{(N-1)\Delta T}\left[\sum_{i=0}^{N-1}\left(\ln\frac{P_{t_{i+1}}}{P_{t_i}}\right)^2-\frac{1}{N}\left(\sum_{i=0}^{N-1}\ln\frac{P_{t_{i+1}}}{P_{t_i}}\right)^2\right].
    \end{split}
\end{equation}
For $\theta$ and $\gamma$, we use Euler-Maruyama scheme and maximum likelihood estimation (MLE).

Discretize
\begin{equation*}
    Z_{t_{i+1}}\approx Z_{t_i}+\theta(P_{t_i}-Z_{t_i})\Delta T+\gamma Z_{t_i}\sqrt{\Delta T}\varepsilon_i,
\end{equation*}
where $\varepsilon_i$'s are independent standard Gaussian.

Conditioned on $P_{t_i}$, $Z_{t_i}$, $\theta$, and $\gamma$, $Z_{t_{i+1}}$ is approximately Gaussian with mean $Z_{t_i}+\theta(P_{t_i}-Z_{t_i})\Delta T$ and variance $\gamma^2Z_{t_i}^2\Delta T$, i.e.
\begin{equation*}
    f(Z_{t_{i+1}}|P_{t_i},Z_{t_i},\theta,\gamma)\approx\frac{1}{\sqrt{2\pi\Delta T}\gamma Z_{t_i}}\exp\left\{-\frac{[Z_{t_{i+1}}-Z_{t_i}-\theta(P_{t_i}-Z_{t_i})\Delta T]^2}{2\gamma^2Z_{t_i}^2\Delta T}\right\}.
\end{equation*}
The minus log likelihood function is
\begin{equation*}
    \begin{split}
        -l(\theta,\gamma)=&\sum_{i=0}^{N-1}\log(\sqrt{2\pi\Delta T}Z_{t_i})+N\log\gamma\\
        &+\sum_{i=0}^{N-1}\frac{[Z_{t_{i+1}}-Z_{t_i}-\theta(P_{t_i}-Z_{t_i})\Delta T]^2}{2\gamma^2Z_{t_i}^2\Delta T}.
    \end{split}
\end{equation*}
Then
\begin{equation*}
    \begin{split}
        &-\frac{\partial l}{\partial\theta}=-\sum_{i=0}^{N-1}\frac{(P_{t_i}-Z_{t_i})[Z_{t_{i+1}}-Z_{t_i}-\theta(P_{t_i}-Z_{t_i})\Delta T]}{\gamma^2Z_{t_i}^2}\\
        &-\frac{\partial l}{\partial\gamma}=\frac{N}{\gamma}-\sum_{i=0}^{N-1}\frac{[Z_{t_{i+1}}-Z_{t_i}-\theta(P_{t_i}-Z_{t_i})\Delta T]^2}{\gamma^3Z_{t_i}^2\Delta T}.
    \end{split}
\end{equation*}
Setting the above to $0$,
\begin{equation}
    \begin{split}
        &\hat{\theta}=\frac{\sum_{i=0}^{N-1}\frac{(Z_{t_{i+1}}-Z_{t_i})(P_{t_i}-Z_{t_i})}{Z_{t_i}^2}}{\Delta T\sum_{i=0}^{N-1}\frac{(P_{t_i}-Z_{t_i})^2}{Z_{t_i}^2}},\\
        &\hat{\gamma}^2=\frac{\sum_{i=0}^{N-1}\frac{(Z_{t_{i+1}}-Z_{t_i})^2}{Z_{t_i}^2}\sum_{i=0}^{N-1}\frac{(P_{t_i}-Z_{t_i})^2}{Z_{t_i}^2}-\left[\sum_{i=0}^{N-1}\frac{(Z_{t_{i+1}}-Z_{t_i})(P_{t_i}-Z_{t_i})}{Z_{t_i}^2}\right]^2}{N\Delta T\sum_{i=0}^{N-1}\frac{(P_{t_i}-Z_{t_i})^2}{Z_{t_i}^2}}.
    \end{split}
\end{equation}
By Cauchy-Schwarz, the estimator $\hat{\gamma}^2$ is non-negative as expected.
\end{appendices}